\documentclass[10pt]{article}

\usepackage{hyperref}
\usepackage{enumerate}
\usepackage{amsmath,amsthm}
\usepackage{fullpage}
\usepackage[all=normal,bibliography=tight]{savetrees}
\usepackage{tikz}
\usetikzlibrary{decorations.markings}
\usetikzlibrary{arrows}
\def\boxit#1{\vbox{\hrule\hbox{\vrule\kern4pt
  \vbox{\kern1pt#1\kern1pt}
\kern2pt\vrule}\hrule}}

\newtheorem{observation}{Observation}
\newtheorem{claim}{Claim}
\newtheorem{lemma}{Lemma}
\newtheorem{theorem}{Theorem}
\newcommand{\Oh}{\mathcal{O}}
\newcommand{\Ohstar}{\Oh^\ast}

\newcommand{\Sdeg}{\mathrm{Ndeg}}
\newcommand{\Gdeg}{\mathrm{Gdeg}_W}
\newcommand{\Tdeg}{\mathrm{Tdeg}}
\usepackage{graphicx}
\usepackage[absolute]{textpos}

\begin{document}

\title{An improved FPT algorithm for Independent Feedback Vertex Set%
\thanks{An extended abstract of this paper has been presented at 44th International Workshop on Graph-Theoretic Concepts in Computer Science~\cite{wg2018}.
  This research is a part of projects that have received funding from the European Research Council (ERC) under the European Union's Horizon 2020 research and innovation programme
under grant agreement No 714704.}}

\author{Shaohua Li\thanks{Institute of Informatics, University of Warsaw, Poland, \texttt{Shaohua.Li@mimuw.edu.pl}.}
  \and
  Marcin Pilipczuk\thanks{Institute of Informatics, University of Warsaw, Poland, \texttt{malcin@mimuw.edu.pl}.}}

  \date{}

\maketitle

\begin{textblock}{20}(0, 12.5)
\includegraphics[width=40px]{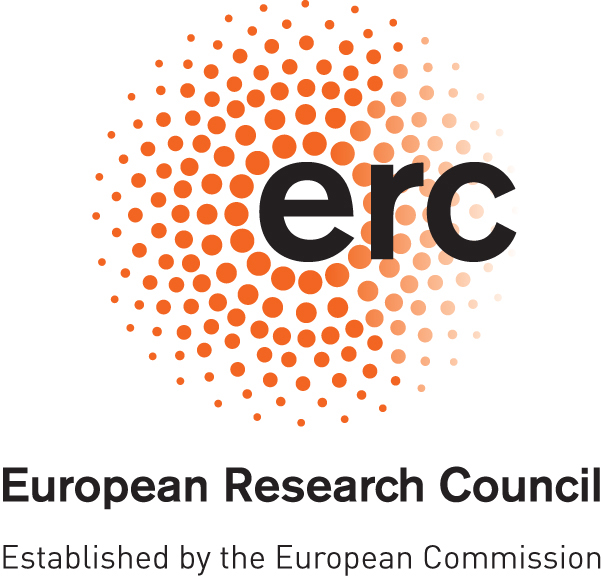}%
\end{textblock}
\begin{textblock}{20}(-0.25, 12.9)
\includegraphics[width=60px]{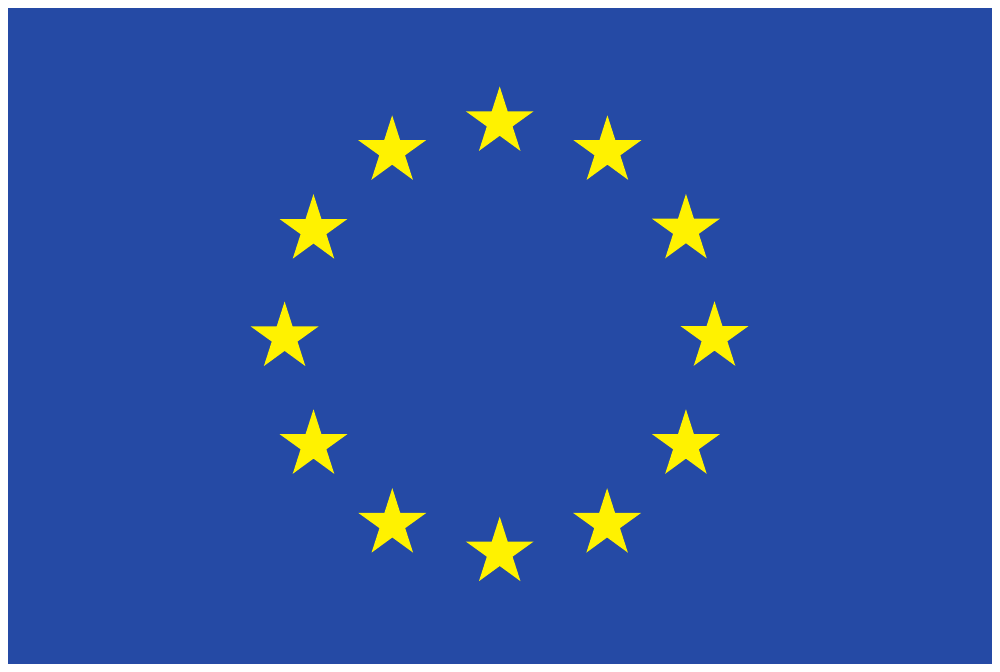}%
\end{textblock}

\begin{abstract}
We study the \textsc{Independent Feedback Vertex Set} problem ---
a variant of the classic \textsc{Feedback Vertex Set}
problem where, given a graph $G$ and an integer $k$,
the problem is to decide whether there exists a vertex set $S\subseteq V(G)$ such that $G\setminus S$ is a forest and $S$ is an independent set of size at most $k$.
We present an $\Ohstar((1+\varphi^{2})^{k})$-time FPT algorithm for this problem,
where $\varphi<1.619$ is the golden ratio, improving the previous fastest $\Ohstar(4.1481^{k})$-time algorithm given by Agrawal et al.~\cite{ipec16}.
The exponential factor in our time complexity bound matches the fastest deterministic FPT algorithm for the classic \textsc{Feedback Vertex Set} problem.

On the technical side, the main novelty is a refined measure of an input instance in a branching
process, that allows for a simpler and more concise description and analysis of the algorithm.

\end{abstract}

\section{Introduction}\label{sec:Introduction}
Given a graph $G$, a feedback vertex set of $G$ is a set of vertices $S\subseteq V(G)$ such that $G\setminus S$ is a forest.
The \textsc{Feedback Vertex Set} problem (\textsc{FVS}) asks to find a feedback vertex set of the minimum size. This problem is a classic NP-hard problem which has been studied extensively in many fields of complexity and algorithms~\cite{opt/2009}.

In this work, we take the point of view of \emph{parameterized complexity}, where
every instance $I$ of a problem at hand is accompanied with a \emph{parameter} $k$,
      intended to represent the complexity of the instance at hand.
      We ask
for a \emph{fixed-parameter algorithm} (\emph{FPT algorithm} for short) that solves
an instance $I$ with parameter $k$ in time $f(k) |I|^c$ for some computable function $f$
and a constant $c$. That is, the exponential blow-up in the running time bound, probably
unavoidable for NP-hard problems, is confined to be a function of the parameter only.
For more on parameterized complexity, we refer to a recent textbook~\cite{pa-book}.

In the context of parameterized complexity of the \textsc{FVS} problem,
there is a long line of work
improving the upper bound of the FPT algorithm for the standard parameterization of
the solution size~\cite{Bodlaender94,CaoC015,ChenFLLV08,DowneyF92,Fellows99,GuoGHNW06,KanjPS04,KociumakaP14} (i.e., the input consists of a graph $G$ and a parameter $k$,
    and the goal is to find a feedback vertex set of size at most $k$ or show that no such set exists).
The fastest randomized FPT algorithm for \textsc{FVS}, which runs in time $\Ohstar(3^{k})$, is given by Cygan et al.~\cite{CyganNPPRW11}.%
\footnote{The $\Ohstar$-notation suppresses factors that are polynomial in the input size.}\footnote{Actually in the randomized FPT algorithm for FVS, the parameter is the treewidth of the graph. Since the treewidth of a yes-instance $(G,k)$ to \textsc{FVS} is at most $k+1$, the randomized algorithm for FVS runs in time $\Ohstar(3^{k})$.}
If one asks for a deterministic FPT algorithm, the champion runs in $O^{*}(3.619^{k})$ and is due to Kociumaka and Pilipczuk~\cite{KociumakaP14}.

At the same time, many variants of \textsc{FVS} received significant attention,
   including \textsc{Subset FVS}~\cite{CyganPPW13,IwataWY16,LokshtanovRS18}, \textsc{Group FVS}~\cite{CyganPP16,Guillemot11a,IwataWY16,KratschW12}, \textsc{Connected FVS}~\cite{MisraPRSS12}, or \textsc{Simultaneous FVS}~\cite{AgrawalLMS18}.

In this paper, we focus on the parameterized version of the \textsc{Independent Feedback Vertex Set} problem (\textsc{IFVS}), which is to decide if there exists a feedback vertex set $S$ of size at most $k$ such that no two vertices of $S$ are adjacent in $G$. Misra et al. gave the first FPT algorithm running in time $\Oh(5^{k}n^{\Oh(1)})$ and an $\Oh(k^{3})$ kernel for \textsc{IFVS}~\cite{MisraPRS12}.
Agrawal et al. presented an improved FPT algorithm running in time $\Ohstar(4.1481^{k})$ for \textsc{IFVS}~\cite{ipec16}. In this paper, we propose a faster FPT algorithm.

\begin{theorem}  \label{main}
The \textsc{Independent Feedback Vertex Set} problem, parameterized by the solution size,
   can be solved in $\Ohstar((1+\varphi^{2})^{k})\leq \Ohstar(3.619^{k})$ time, where $\varphi=\frac{1+\sqrt{5}}{2}<1.619$ is the golden ratio.
\end{theorem}
We remark here that Theorem~\ref{main} is not ``just another'' improvement in the base of the
exponential function, but in some sense ``the end of the road''. The exponential
function of the time bound of Theorem~\ref{main} matches the one of the algorithm
of Kociumaka and Pilipczuk~\cite{KociumakaP14} for the classic \textsc{FVS} problem.
Since \textsc{FVS} trivially reduces to \textsc{IFVS}
(subdivide each edge once), any (deterministic) improvement to the base of the exponential
function of Theorem~\ref{main} would give a similar improvement for \textsc{FVS}.

On the technical side, we follow the standard approach of iterative compression as in Agrawal et al.~\cite{ipec16}
to reduce to a ``disjoint'' version of the problem.
Here, our approach diverges from the one of Agrawal et al~\cite{ipec16}.
We follow a modified measure for the subsequent branching process, somewhat inspired by
the work of Kociumaka and Pilipczuk~\cite{KociumakaP14}.
This improved measure, together with a number of new notions (generalized $W$-degree,
    potential nice vertices and tents), allow us to simplify the algorithm and analysis
as compared to~\cite{ipec16}.

\section{Preliminaries}\label{sec:Preliminaries}
The graphs in our paper are all undirected and may contain multiple edges or loops.
For a graph $G$, we denote its vertex set by $V(G)$ and edge multiset by $E(G)$. For a vertex $v\in V(G)$, we use $N(v)=\{u\in V(G): uv\in E(G)\}$ to denote the \emph{neighborhood} of $v$;
note that $N(v)$ is a set, containing a vertex $u$ only once even in the presence of multiple edges $uv$. We define the \emph{closed neighborhood} of $v$ as $N[v]=N(v)\cup \{v\}$. For a vertex set $A\subseteq V(G)$, the neighborhood of $A$ is $N(A)=\bigcup_{v\in A}N(v) \setminus A$.  For a vertex set $X\subseteq V(G)$, we denote the \emph{induced subgraph} of $X$ by $G[X]$. For simplicity, we use $G\setminus X$ to denote $G[V(G)\setminus X]$. For a vertex set $X\subseteq V(G)$ and $v\in V(G)$, we define \emph{$X$-degree} of $v$ as the number of edges with one endpoint being $v$
and the other lying in $X$, and we denote it by $\deg_{X}(v)$. Note that the $X$-degree counts edges with multiplicities.
A \emph{connected component} is a maximal connected subgraph. Contracting a connected subgraph $H$ is the operation of replacing the subgraph $H$ with a vertex $v_H$
and every edge $xy$ with $x \in V(H)$ and $y \in V(G) \setminus V(H)$ with an edge $v_Hy$ (keeping multiplicities).

In the realms of parameterized complexity, every instance $I$ of a problem at hand is accompanied with a parameter $k$.
A \emph{fixed-parameter algorithm} (\emph{FPT algorithm} for short) solves an instance $I$ with parameter $k$ in time $f(k) |I|^c$ for some computable function $f$
and a constant $c$.
A \emph{kernel} of size $g(k)$ for some computable function $g$ is a polynomial-time procedure that reduces an instance $I$ with parameter $k$ to an equivalent instance with size and parameter value bounded by $g(k)$.

\section{An Algorithm for \textsc{Independent Feedback Vertex Set}}

Given an instance $(G,k)$, we first invoke the $\Ohstar((1+\varphi^2)^k)$-time FPT algorithm for
the classic \textsc{FVS} problem~\cite{KociumakaP14}.
If the algorithm returns NO, we conclude that there is no independent feedback vertex set of size at most $k$ since an independent feedback vertex set is also a feedback vertex set.
Otherwise, the algorithm returns a feedback vertex set $Z$ such that $|Z|\leq k$. Obviously, $F=G\setminus Z$ is a forest.

Suppose there is a solution $S$ for the input instance $(G,k)$.
The algorithm branches into $2^{|Z|}$ directions,
guessing a subset $Z'$ of $Z$ such that $S\cap Z=Z'$.
Let $W = Z \setminus Z'$.
If $G[Z']$ is not an independent set or $G[W]$ is not a forest, the algorithm rejects this guess.
Hence, we can assume that $G[Z']$ is an independent set and $G[W]$ is a forest.
Let $R=N(Z')\cap V(F)$.
Since the solution $S$ is an independent set and $Z' \subseteq S$, we have $R \cap S = \emptyset$.
Then the algorithm tries to find an independent feedback vertex set $S'\subseteq V(F)$ for $G\setminus Z'$ such that $S'\cap R=\emptyset$ and $|S'|\leq k-|Z'|$.
Note that at this point we cannot delete the vertices from $W$ nor $R$ from the graph as, albeit undeletable, they take part in the structure of cycles in $G$ that we are to break.
Following Agrawal et al.~\cite{ipec16}, we call this subproblem \textsc{Disjoint Independent Feedback Vertex Set} (\textsc{DIS-IFVS} for short).
We give a faster FPT algorithm for \textsc{DIS-IFVS} in the next section.
The algorithm tries every possible $Z'\subseteq Z$ and solves the corresponding subproblem of
\textsc{DIS-IFVS}. If the algorithm finds a YES instance of \textsc{DIS-IFVS}, then it returns YES for the instance $(G,k)$ of \textsc{IFVS}. Otherwise, if the algorithm tries every possible $Z'\subseteq Z$ and obtains a NO answer for every corresponding instance of \textsc{DIS-IFVS}, it reports that $(G,k)$ is a NO instance.

\subsection{\textsc{Disjoint Independent Feedback Vertex Set}}\label{sec:DIS-IFVS}

We start with a formal definition of the problem.

\begin{quote}
\textsc{Disjoint Independent Feedback Vertex Set}\\
\textbf{Input:} An undirected (multi)graph $G$, a feedback vertex set $W$ of $G$, $R\subseteq V(G)\setminus W$, and an integer $k$.  \\
\textbf{Question:} Is there an independent feedback vertex set $X\subseteq V(G)\setminus(W\cup R)$ for $G$ such that $|X|\leq k$?\\
\end{quote}
Let $F=V(G\setminus W)$. Obviously, $G[F]$ is a forest since $W$ is a feedback vertex set of $G$.
A vertex $v\in F \setminus R$ is a \emph{nice vertex} if $\deg_W(v)=2$ and $v$ has no neighbors in $F$. A vertex $v\in F \setminus R$ is a \emph{tent} if $\deg_W(v)=3$ and $v$ has no neighbors in $F$.

As mentioned earlier, we rely on a measure different from the one in \cite{ipec16}.
The measure $\mu$ of an instance $(G,W,R,k)$ is defined as
$$\mu=k+\rho-(\eta+\tau).$$
Here, $\rho$ represents the number of connected components of $G[W]$,
 $\eta$ is the number of nice vertices in $F \setminus R $ and $\tau$ is the number of tents in $F \setminus R$.

We remark that the distinction between sets $W$ and $R$ is purely for the sake of complexity of the algorithm. The set of feasible solutions to a \textsc{DIS-IFVS}
instance $(G, W, R, k)$ would be the same if we move vertices from $R$ to $W$. However, the notions of tents, nice vertices, and the measure $\mu$ strongly depends on the distinction between
the sets $W$ and $R$. The algorithm maintains this distinction to ensure the promised running time bound.

Our main technical result is the following.
\begin{lemma}\label{technical}
A \textsc{Disjoint Independent Feedback Vertex Set} instance $I$ with measure $\mu$
can be solved in time $\Ohstar(\varphi^{\mu})$, where $\varphi=\frac{1+\sqrt{5}}{2}$ is the golden ratio.
\end{lemma}
Theorem~\ref{main} follows by standard analysis as in~\cite{ipec16}:
\begin{proof}[Proof of Theorem~\ref{main}]
The algorithm for \textsc{FVS} of~\cite{KociumakaP14} runs in time $\Ohstar((1+\varphi^2)^{k})$.
In a branch with a set $Z' \subseteq Z$ the routine for \textsc{DIS-IFVS} is passed
an instance with both $W=Z\setminus Z'$ and the parameter bounded by $k-|Z'|$, and hence with
measure bounded by $2(k-|Z'|)$.
Since the algorithm for \textsc{DIS-IFVS} runs in time $O^{*}(\varphi^{\mu})$,
the total running time of its applications
is bounded by
$$\sum_{i=0}^k \binom{k}{i} \Ohstar(\varphi^{2(k-i)}) = \Ohstar((1+\varphi^{2})^{k}) \leq \Ohstar(3.619^{k}).$$
This completes the proof.
\end{proof}

The remainder of this section is devoted to the proof of Lemma~\ref{technical}.
We start with showing that $\mu$ is nonnegative on YES instances.
\begin{lemma} \label{measure}
Let $I=(G,W,R,k)$ be a YES instance of \textsc{DIS-IFVS}.
Then $\mu \geq 0$.
\end{lemma}
\begin{proof}
Let $X$ be a solution to the instance $I$. Thus $G'=G\setminus X$ is a forest. Let $N\subseteq V(G)\setminus (W\cup R)$ be the set of nice vertices and $T\subseteq V(G)\setminus (W\cup R)$ be the set of tents. Since $X \cap W = \emptyset$, we have that $H := G[W\cup (N \setminus X) \cup (T \setminus X)]$ is a forest.
Now we contract each component in $H[W]$ into a single vertex and get a forest $\tilde{H}$. Since there are at most $\rho+|N\setminus X|+|T\setminus X|$ vertices in $\tilde{H}$, there are at most $\rho+|N\setminus X|+|T\setminus X|-1$ edges in $\tilde{H}$. According to the definition of tents and nice vertices, $(N\cup T)\setminus X$ is an independent set. Moreover, since the degree of any vertex in $N \setminus X$ and $T \setminus X$ is 2 and 3, respectively, we get the following inequality:
$$2|N\setminus X|+3|T\setminus X|\leq |E(\tilde{H})|\leq \rho+|N\setminus X|+|T\setminus X|-1.$$
It follows that:
$$|N\setminus X|+|T\setminus X|\leq |N\setminus X|+2|T\setminus X| \leq \rho.$$
Hence, as $|X| \leq k$,
$$|N|+|T| \leq \rho+k.$$
As a result, $\mu=\rho+k-(\eta+\tau) \geq 0$.
\end{proof}
A small comment is in place. Our measure $\mu$ is different from the one of~\cite{ipec16}:
$\mu'=2k+\rho-(\eta+2\tau)$. The change in the measure is one of the critical insights
in this paper: while it sometimes leads to weaker branching vectors as compared
to~\cite{ipec16}, the ``starting value'' in an application in the
above proof of Theorem~\ref{main} is $2(k-|Z'|)$, not $3(k-|Z'|)$ as in~\cite{ipec16}.
Thus, to obtain the promised running time bound, we are fine with branching vectors
of the form $(1,2)$; that is, we are fine with branching steps in two directions,
where in one direction the measure drops by at least one, and in the other direction
by at least two.
The change in the measure is similar to the one that happened in the work of Kociumaka
and Pilipczuk for \textsc{FVS}~\cite{KociumakaP14},
as compared to a previous champion of Cao, Chen, and Liu~\cite{CaoC015}.

\medskip

We introduce now some definitions that will help us streamline later arguments.
Let $(G,W,R,k)$ be an instance of \textsc{DIS-IFVS} and let $F=V(G)\setminus W$.
We say that $u\in F \setminus R$ is a \emph{potential nice vertex} or \emph{P-nice} if
$u$ is of degree $2$ and exactly one of its incident edges has a second endpoint in $W$.
For a vertex $v$ in $G[F]$, we define the \emph{nice degree} of $v$, denoted by $\Sdeg(v)$,
   as the number of P-nice neighbors of $v$.
A \emph{generalized degree} of $v$ is $\Gdeg(v)=\Sdeg(v)+\deg_{W}(v)$.
We say that $u\in F \setminus R$ is a \emph{potential tent} or \emph{P-tent} if $\Gdeg(u)=2$
and $\deg(u) = 3$.
For a vertex $v$ in $F$, we define the \emph{tent degree}
of $v$, denoted by $\Tdeg(v)$, as the number of neighbors of $v$ that are P-tents.
See Figure~\ref{pnice} for an illustration.

\begin{figure}[htb]
\begin{center}
\begin{tikzpicture}[scale=0.7]
    \draw(0,4) ellipse (4cm and 2cm);
    \draw (0,0) ellipse (4cm and 1cm);
    \node at (1,5) {\Large $F$ };
    \node (v6) at (0,-0.5){\Large $W$};
     \node (v7) at (-3,5.55) {};
     \node (v8) at (-3,2.45) {};
     \draw[bend left=20] (v7) to (v8);

      \node at (-3.3,4.5) {\Large R};

      \node[style={fill=black,circle,draw,inner sep=1pt}] (v10) at (-2,0.3) {};
      \node[style={fill=black,circle,draw,inner sep=1pt}] (v11) at (-3.2,0.3) {};
      \node[style={fill=black,circle,draw,inner sep=1pt}] (v101) at (-2,2.8) {};
      \draw (v101) to (v10);
      \draw (v101) to (v11);
      \node[style={fill=black,circle,draw,inner sep=1pt}] (v12) at (-1.5,3) {};
      \node[style={fill=black,circle,draw,inner sep=1pt}] (v13) at (-1,0.3) {};
      \node[style={fill=black,circle,draw,inner sep=1pt}] (v14) at (0,0.3) {};
      \draw (v10) to (v12);
      \draw (v13) to (v12);
      \draw (v14) to (v12);
      \node (nice) at (-2.2,3.8) {nice};
      \draw [->,red] (-2.2,3.5) to (-2,3);
      \node (tent) at (-1.2,4.3) {tent};
      \draw [->,red] (-1.3,4.0) to (-1.5,3.2);
      
      \node[style={fill=black,circle,draw,inner sep=1pt}] (x2) at (0.6,0.3) {};
      \node[style={fill=black,circle,draw,inner sep=1pt}] (x1) at (0.3,2.8) {};
      \node[style={fill=black,circle,draw,inner sep=1pt}] (x3) at (-0.5,3.0) {};
      \draw (v14) to (x1);
      \draw (x2) to (x1);
      \draw (x3) to (x1);
      \node[style={fill=black,circle,draw,inner sep=1pt}] (y2) at (1.6,0.3) {};
      \node[style={fill=black,circle,draw,inner sep=1pt}] (y1) at (1.6,2.8) {};
      \node[style={fill=black,circle,draw,inner sep=1pt}] (y3) at (2.2,3.0) {};
      \draw (y2) to (y1);
      \draw (y3) to (y1);
      
      \node (nice) at (0.2,3.8) {$P$-tent};
      \draw [->,red] (0.2,3.5) to (0.3,3);
      \node (tent) at (2.0,4) {$P$-nice};
      \draw [->,red] (2.0,3.7) to (1.7,3.0);
\end{tikzpicture}
\end{center}
\caption{A nice vertex, a tent, a $P$-tent, and a $P$-nice vertex.} \label{pnice}
\end{figure}
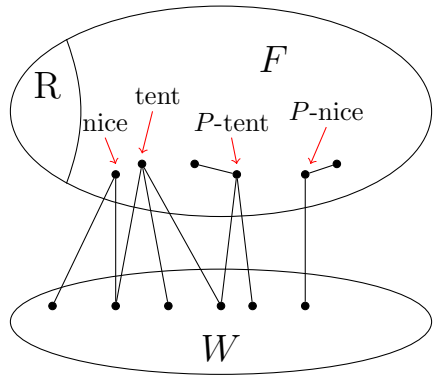

\subsection{Reduction Rules for \textsc{DIS-IFVS}}\label{sec:Reduction}

Now we introduce some reduction rules for DIS-IFVS. We always apply the applicable reduction rule of the lowest number. First, let us introduce five reduction rules from \cite{ipec16}.\\

\noindent \textbf{Reduction Rule 1}: Delete any vertex of degree at most one. \\

\noindent \textbf{Reduction Rule 2}: Let $u$, $v$ be two adjacent vertices of $G \setminus W$ that are both not nice and both of degree $2$ in $G$.
Let $x$ be the second neighbor of $u$ and let $y$ be the second neighbor of $v$ ($x$ and $y$ could be the same vertex).
If neither $u$ nor $v$ is in $R$ or both are in $R$,
then delete one vertex in $\{u,v\}$ arbitrarily, say $u$, and connect the two neighbors $u$ (i.e., $v$ and the other neighbor of $u$) with a new edge.
If exactly one of $u$ and $v$ is in $R$, say $v\in R$, then delete $v$ and add an edge between its neighbors (i.e., an edge $uy$).\\

\noindent \textbf{Reduction Rule 3}: If $k<0$ or $\mu < 0$,
  return that the input instance is a NO instance. \\

\noindent \textbf{Reduction Rule 4}: If there is a vertex $v\in R$ such that $v$ has two incident edges with the second endpoints in the same component of $W$, then return that the input instance is a NO instance.  \\

\noindent \textbf{Reduction Rule 5}: If there is a vertex $v\in F\setminus R$ such that $v$ has at least two incident edges with the second endpoints in the same component of $W$, then remove $v$ from $G$ and add all vertices in $F\cap N(v)$ to $R$. In this case, $k$ decreases by one.\\

It is not difficult to verify the safeness of Reduction Rules $1-5$ as shown in \cite{ipec16}. But when analyzing Reduction Rules $1$ and $5$,
we need to be careful since we use a different measure $\mu=k+\rho-(\eta+\tau)$.
In Reduction Rule 1, if one deletes a neighbor $w$ of a tent or a nice vertex $v$, then $v$ stops being a tent or a nice vertex ($\eta+\tau$ could decrease by one),
but also $\{w\}$ stops being a connected component of $G[W]$ (decreasing $\rho$ by one).
For Reduction Rule 5, it may happen that $v$ is a tent or a nice vertex, and its deletion decreases $\eta+\tau$ by one.
However, the removal of $v$ also decreases $k$ by one. Thus $\mu$ does not increase.\\

Now we introduce two new reduction rules.

\medskip

\noindent \textbf{Reduction Rule 6}: If there is a vertex $v\in R$ such that $\Gdeg(v)\geq 1$ or $\Tdeg(v)\geq 1$, then remove $v$ from $R$ and add $v$ to $W$. \\

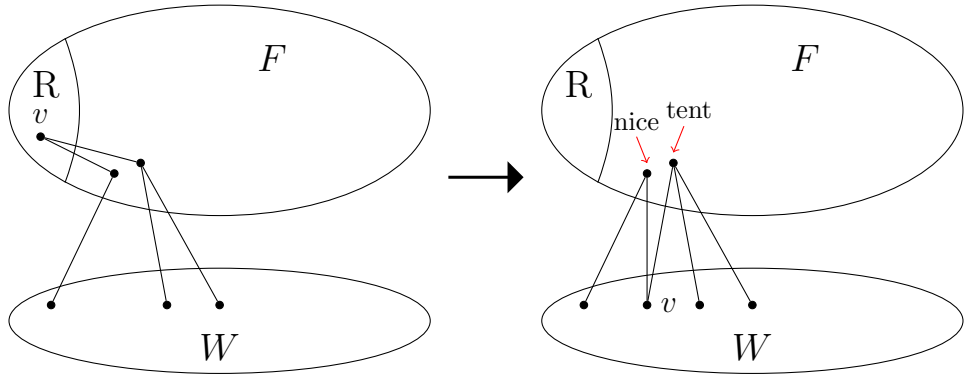
\begin{figure}[htb]
\begin{center}
\begin{tikzpicture}[scale=0.7]
    \draw(0,4) ellipse (4cm and 2cm);
    \draw (0,0) ellipse (4cm and 1cm);
    \node at (1,5) {\Large $F$ };
    \node (v6) at (0,-0.5){\Large $W$};
     \node (v7) at (-3,5.55) {};
     \node (v8) at (-3,2.45) {};
     \draw[bend left=20] (v7) to (v8);

      \node at (-3.3,4.5) {\Large R};

      \node[label=above:{\large $v$},style={fill=black,circle,draw,inner sep=1pt}] (v10) at (-3.4,3.5) {};
      \node[style={fill=black,circle,draw,inner sep=1pt}] (v101) at (-2,2.8) {};
      \node[style={fill=black,circle,draw,inner sep=1pt}] (v11) at (-3.2,0.3) {};
      \draw (v101) to (v10);
      \draw (v101) to (v11);
      \node[style={fill=black,circle,draw,inner sep=1pt}] (v12) at (-1.5,3) {};
      \node[style={fill=black,circle,draw,inner sep=1pt}] (v13) at (-1,0.3) {};
      \node[style={fill=black,circle,draw,inner sep=1pt}] (v14) at (0,0.3) {};
      \draw (v10) to (v12);
      \draw (v13) to (v12);
      \draw (v14) to (v12);
\end{tikzpicture}
\begin{tikzpicture}
\node (s1) at (0,0) {};
\node (s2) at (1,0) {};
\tikzset{>=triangle 90}
\draw [->,line width=0.4mm] (0,2.5)-- (1,2.5);
\end{tikzpicture}
\begin{tikzpicture}[scale=0.7]
    \draw(0,4) ellipse (4cm and 2cm);
    \draw (0,0) ellipse (4cm and 1cm);
    \node at (1,5) {\Large $F$ };
    \node (v6) at (0,-0.5){\Large $W$};
     \node (v7) at (-3,5.55) {};
     \node (v8) at (-3,2.45) {};
     \draw[bend left=20] (v7) to (v8);

      \node at (-3.3,4.5) {\Large R};

      \node[label=right:{\large $v$},style={fill=black,circle,draw,inner sep=1pt}] (v10) at (-2,0.3) {};
      \node[style={fill=black,circle,draw,inner sep=1pt}] (v11) at (-3.2,0.3) {};
      \node[style={fill=black,circle,draw,inner sep=1pt}] (v101) at (-2,2.8) {};
      \draw (v101) to (v10);
      \draw (v101) to (v11);
      \node[style={fill=black,circle,draw,inner sep=1pt}] (v12) at (-1.5,3) {};
      \node[style={fill=black,circle,draw,inner sep=1pt}] (v13) at (-1,0.3) {};
      \node[style={fill=black,circle,draw,inner sep=1pt}] (v14) at (0,0.3) {};
      \draw (v10) to (v12);
      \draw (v13) to (v12);
      \draw (v14) to (v12);
      \node (nice) at (-2.2,3.8) {nice};
      \draw [->,red] (-2.2,3.5) to (-2,3);
      \node (tent) at (-1.2,4) {tent};
      \draw [->,red] (-1.3,3.7) to (-1.5,3.2);
\end{tikzpicture}
\end{center}
\caption{Reduction Rule $6$} \label{Reduction6}
\end{figure}

\noindent \textbf{Reduction Rule 7}: If there is a vertex $v\in F\setminus R$
such that every neighbor $w \in N(v) \setminus (W \cup R)$ is of degree $2$,
and at least one such neighbor exists, then put $N(v) \setminus (W \cup R)$ into $R$.

\begin{figure}[htb]
\begin{center}
\begin{tikzpicture}[scale=0.7]
    \draw(0,4) ellipse (4cm and 2cm);
    \draw (0,0) ellipse (4cm and 1cm);
    \node at (3,4) {\Large $F$ };
    \node (v6) at (0,-0.5){\Large $W$};
     \node (v7) at (-1,6.1) {};
     \node (v8) at (-1,1.9) {};
     \draw[bend left=20] (v7) to (v8);

      \node at (-3.5,4) {\Large R};

      \node[label=right:{\large $v$},style={fill=black,circle,draw,inner sep=1pt}] (v11) at (0,4) {};
      \node[label=below:{\large $w_{3}$},style={fill=black,circle,draw,inner sep=1pt}] (w3) at (-1,4.3) {};
      \node[style={fill=black,circle,draw,inner sep=1pt}] (w31) at (-2.5,4) {};
      \node[style={fill=black,circle,draw,inner sep=1pt}] (w311) at (-3,4.5) {};

      \node[style={fill=black,circle,draw,inner sep=1pt}] (w313) at (-3.2,4) {};
      \node[style={fill=black,circle,draw,inner sep=1pt}] (w32) at (-2,4.8) {};
      \draw (v11) to (w3);
      \draw (w3) to (w31);
      \draw (w31) to (w311);

      \draw (w31) to (w313);
      \draw (w3) to (w32);

      \node[label=above:{\large $w_{4}$},style={fill=black,circle,draw,inner sep=1pt}] (w4) at (-1.2,5) {};
      \node[style={fill=black,circle,draw,inner sep=1pt}] (w41) at (0,5.2) {};
      \node[style={fill=black,circle,draw,inner sep=1pt}] (w411) at (0.6,5.8) {};

      \node[style={fill=black,circle,draw,inner sep=1pt}] (w413) at (1,4.8) {};
      \draw (v11) to (w4);
      \draw (w4) to (w41);
      \draw (w41) to (w411);

      \draw (w41) to (w413);

      \node[label=left:{\large $w_{1}$},style={fill=black,circle,draw,inner sep=1pt}] (v12) at (0.5,3) {};

      \node[style={fill=black,circle,draw,inner sep=1pt}] (v14) at (1,0.3) {};
      \node[style={fill=black,circle,draw,inner sep=1pt}] (v15) at (0,0.3) {};

      \node[label=right:{\large $w_{2}$},style={fill=black,circle,draw,inner sep=1pt}] (v17) at (1.5,3) {};
      \draw (v11) to (v12);
      \draw (v11) to (v17);
      \draw (v15) to (v12);
      \draw (v14) to (v17);
\end{tikzpicture}
\begin{tikzpicture}
\node (s1) at (0,0) {};
\node (s2) at (1,0) {};
\tikzset{>=triangle 90}
\draw [->,line width=0.4mm] (0,2.5)-- (1,2.5);
\end{tikzpicture}
\begin{tikzpicture}[scale=0.7]
    \draw(0,4) ellipse (4cm and 2cm);
    \draw (0,0) ellipse (4cm and 1cm);
    \node at (3,4) {\Large $F$ };
    \node (v6) at (0,-0.5){\Large $W$};
     \node (v7) at (-1,6.1) {};
     \node (v8) at (-1,1.9) {};
     \draw[bend left=20] (v7) to (v8);

      \node at (-3.5,4) {\Large R};


      \node[label=right:{\large $v$},style={fill=black,circle,draw,inner sep=1pt}] (v11) at (0,4) {};
      \node[label=below:{\large $w_{3}$},style={fill=black,circle,draw,inner sep=1pt}] (w3) at (-1,4.3) {};
      \node[style={fill=black,circle,draw,inner sep=1pt}] (w31) at (-2.5,4) {};
      \node[style={fill=black,circle,draw,inner sep=1pt}] (w311) at (-3,4.5) {};

      \node[style={fill=black,circle,draw,inner sep=1pt}] (w313) at (-3.2,4) {};
      \node[style={fill=black,circle,draw,inner sep=1pt}] (w32) at (-2,4.8) {};
      \draw (v11) to (w3);
      \draw (w3) to (w31);
      \draw (w31) to (w311);

      \draw (w31) to (w313);
      \draw (w3) to (w32);

      \node[label=above:{\large $w_{4}$},style={fill=black,circle,draw,inner sep=1pt}] (w4) at (-1.2,5) {};
      \node[style={fill=black,circle,draw,inner sep=1pt}] (w41) at (0,5.2) {};
      \node[style={fill=black,circle,draw,inner sep=1pt}] (w411) at (0.6,5.8) {};

      \node[style={fill=black,circle,draw,inner sep=1pt}] (w413) at (1,4.8) {};
      \draw (v11) to (w4);
      \draw (w4) to (w41);
      \draw (w41) to (w411);

      \draw (w41) to (w413);

      \node[label=left:{\large $w_{1}$},style={fill=black,circle,draw,inner sep=1pt}] (v12) at (-2,3) {};

      \node[style={fill=black,circle,draw,inner sep=1pt}] (v14) at (1,0.3) {};
      \node[style={fill=black,circle,draw,inner sep=1pt}] (v15) at (0,0.3) {};

      \node[label={[xshift=5mm, yshift=-3mm]\large $w_{2}$},style={fill=black,circle,draw,inner sep=1pt}] (v17) at (-1,3) {};
      \draw (v11) to (v12);
      \draw (v11) to (v17);
      \draw (v15) to (v12);
      \draw (v14) to (v17);
\end{tikzpicture}
\end{center}
\caption{Reduction Rule $7$} \label{Reduction7}
\end{figure}
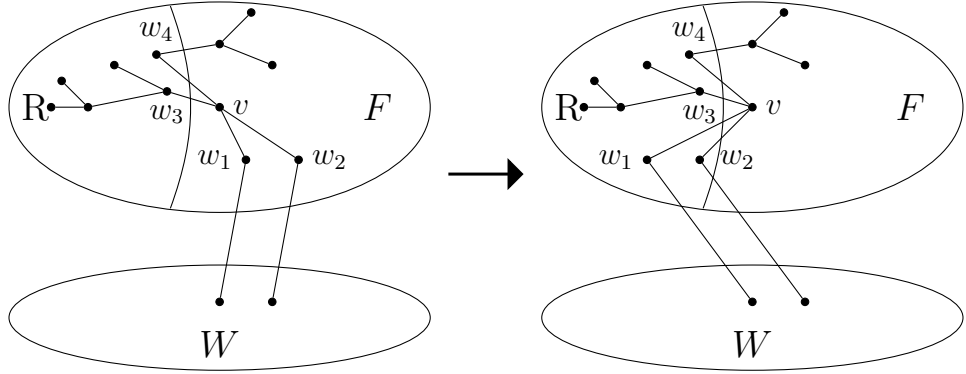
\medskip

We first show their safeness.

\begin{claim}
Reduction Rules $6$ and $7$ are safe.
\end{claim}
\begin{proof}
The safeness of Reduction Rule 6 is straightforward.
For the safeness of Reduction Rule $7$, suppose that $(G,W,R,k)$ is an input instance.
Let $v$ be the vertex satisfying the condition of Reduction Rule $7$ and $(G,W,R\cup (N(v)\cap F),k)$ be the instance obtained after applying Reduction Rule $7$. We claim that $(G,W,R,k)$ is a YES instance if and only if $(G,W,R\cup(N(v) \cap F),k)$ is a YES instance.
The ``if'' direction is straightforward, since we only increased the set $R$.

For the ``only if'' direction, let $X$ be a solution of size at most $k$ to the instance $(G,W,R,k)$.
If $X\cap N(v)=\emptyset$, then $X$ is also a solution to $(G,W,R\cup(N(v)\cap F),k)$. Otherwise, we construct a vertex set $X'=(X\cup \{v\})\setminus (N(v)\cap F)$. Obviously $|X'|\leq k$. We will show that $X'$ is a solution to $(G,W,R\cup (N(v)\cap F),k)$. Clearly, it is disjoint with $W \cup R \cup N((v) \cap F)$ and independent, as it is disjoint with $N(v)$.
To show that $X'$ is a feedback vertex set in $G$,
observe that since every vertex $w \in N(v) \setminus (W \cup R)$ is of degree $2$,
  every cycle passing through $w$ in $G$ passes also through $v$.
 \end{proof}

Since Reduction Rule 7 only moves vertices to $R$, its application does not change
the measure; note that the neighbors of a vertex affected by Reduction Rule 7 can be neither a nice vertex nor a tent.
However, the situation is not that easy for Reduction Rule 6,
and we need to show that its application does not increase $\mu$.
To this end, we show a number of generic observations on how the measure $\mu$ changes
if we modify a neighbor of a P-nice vertex or a P-tent; we refer to Figure~\ref{Obs1-3}
for an illustration.

\begin{observation}\label{obs:nice-W}
Let $v \in F$ be a vertex with a P-nice neighbor $w$.
Consider the operation of moving $v$ to $W$.
Then, the vertex $w$ becomes nice and $\eta$ goes up at least by one.
\end{observation}

\begin{observation}\label{obs:tent-X}
Let $v \in F$ be a vertex with a P-tent neighbor $w$ such that $v$ is not P-nice.
Consider the operation of putting $v$ in a solution: deleting it from $G$
and putting $N(v) \cap F$ into $R$.
Then the application of reduction rules on $w$ and its (possible) other neighbors in $F$
decreases $\mu$ by at least one.
\end{observation}
\begin{proof}
The operation moves $w$ to $R$ and decreases its degree to $2$.
Since $w$ is a P-tent and $v$ is not a P-nice vertex, every neighbor $u \in (N(w) \cap F)\setminus \{v\}$ is a P-nice vertex.
Consequently, Reduction Rule 2 reduces $(N[w] \cap F) \setminus \{v\}$ to a single vertex $w'$,
  which is in $R$ if $(N(w) \cap F) \setminus \{v\} \subseteq R$.
  Furthermore, $\deg(w') = \deg_W(w') = 2$.
We remark here that the above discussion includes the case when $N(v) \cap F = \{v\}$,
   that is, all other neighbors of $v$ are already in $W$.

If $w'$ has both neighbors in the same connected component of $G[W]$, then
either Reduction Rule 4 rejects the instance or Reduction Rule 5 decreases $k$ by one.
Otherwise, if $w' \in R$, Reduction Rule 6 moves $w'$ to $W$, decreasing $\rho$ by one.
If $w' \notin R$, then $w'$ becomes a nice vertex, increasing $\eta$ by one.
Thus, in all cases, $\mu$ decreases by at least one.
 \end{proof}

\begin{observation}\label{obs:tent-W}
Let $v \in F$ be a vertex with a P-tent neighbor $w$ such that $v$ is not P-nice.
Consider the operation of moving $v$ into $W$.
Then the application of reduction rules on $w$ and its (possible) other neighbors in $F$
decrease $\mu$ by at least one.
\end{observation}
\begin{proof}
Since $w$ is a P-tent and $v$ is not P-nice, every neighbor $u \in (N(w) \cap F) \setminus \{v\}$ is P-nice.
Consider such a vertex $u$; note that $u \in F \setminus R$ by the definition of P-nice.
Reduction Rule 7 is applicable to $w$; this rule would move $u$ to $R$
and then Reduction Rule 6 would move $u$ to $W$ and, since $u$ is P-nice,
this would not create a new connected component of $G[W]$.
Along this process, Reduction Rule 5 can be triggered on $w$, deleting $w$
and decreasing $k$ by one.

If this does not happen, in the end of this process, we have $\deg(w) = \deg_W(w) = 3$;
note that we are already in this situation if $N(w) \cap F = \{v\}$ in the beginning. 
Since $w$ is a P-tent at the beginning, $w \notin R$ in the end of the process.
Then, $w$ becomes a tent and increases $\tau$ by one.
Thus, in all cases, $\mu$ decreases by at least one.  \end{proof}

We remark here that Observations~\ref{obs:tent-X} and~\ref{obs:tent-W} 
treat measure drop \emph{after} the respective operation on $v$ is applied;
it does not count how the operation on $v$ itself affects the measure.

\begin{figure}[htb]
\begin{center}
\begin{tikzpicture}[scale=0.5]

    \draw(0,4) ellipse (4cm and 2cm);
    \draw (0,0) ellipse (4cm and 1cm);
    \node at (3.5,4) {\Large $F$ };
    \node (v6) at (0,-0.5){\Large $W$};
     \node (v7) at (-1,6.2) {};
     \node (v8) at (-1,1.8) {};
     \draw[bend left=20] (v7) to (v8);
     \node at (-3.5,4) {\Large R};
     \node[label=above:{$v$},style={fill=black,circle,draw,inner sep=1pt}] (v) at (1.5,5) {};
     \node[label=right:{$w$},style={fill=black,circle,draw,inner sep=1pt}] (w1) at (2,4.2) {};
     \node[style={fill=black,circle,draw,inner sep=1pt}] (w11) at (1.8,3) {};
     \node[style={fill=black,circle,draw,inner sep=1pt}] (w12) at (2.5,3.5) {};
     \node[style={fill=black,circle,draw,inner sep=1pt}] (u11) at (1.5,0) {};
     \node[style={fill=black,circle,draw,inner sep=1pt}] (u12) at (2.5,0) {};
     \node[label=right:{$w_{2}$},style={fill=black,circle,draw,inner sep=1pt}] (w2) at (0,4) {};
     \node[style={fill=black,circle,draw,inner sep=1pt}] (u21) at (-1,0) {};
      \draw (v) to (w1);
      \draw (w1) to (w11);
      \draw (w1) to (w12);
      \draw (w11) to (u11);
      \draw (w12) to (u12);
      \draw (v) to (w2);
      \draw (w2) to (u21);

\begin{scope}[shift={(12cm,5cm)}]
    \draw(0,4) ellipse (4cm and 2cm);
    \draw (0,0) ellipse (4cm and 1cm);
    \node at (3.5,4) {\Large $F$ };
    \node (v6) at (0,-0.5){\Large $W$};
     \node (v7) at (-1,6.2) {};
     \node (v8) at (-1,1.8) {};
     \draw[bend left=20] (v7) to (v8);
     \node at (-3.5,4) {\Large R};

     \node[label=right:{$w$},style={fill=black,circle,draw,inner sep=1pt}] (w1) at (-2,5) {};
     \node[style={fill=black,circle,draw,inner sep=1pt}] (w11) at (1.8,3) {};
     \node[style={fill=black,circle,draw,inner sep=1pt}] (w12) at (2.5,3.5) {};
     \node[style={fill=black,circle,draw,inner sep=1pt}] (u11) at (1.5,0) {};
     \node[style={fill=black,circle,draw,inner sep=1pt}] (u12) at (2.5,0) {};
     \node[label=right:{$w_{2}$},style={fill=black,circle,draw,inner sep=1pt}] (w2) at (-2,3.5) {};
     \node[style={fill=black,circle,draw,inner sep=1pt}] (u21) at (-1,0) {};

      \draw (w1) to (w11);
      \draw (w1) to (w12);
      \draw (w11) to (u11);
      \draw (w12) to (u12);

      \draw (w2) to (u21);
\end{scope}

\begin{scope}[shift={(24cm,5cm)}]
    \draw(0,4) ellipse (4cm and 2cm);
    \draw (0,0) ellipse (4cm and 1cm);
    \node at (3.5,4) {\Large $F$ };
    \node (v6) at (0,-0.5){\Large $W$};
     \node (v7) at (-1,6.2) {};
     \node (v8) at (-1,1.8) {};
     \draw[bend left=20] (v7) to (v8);
      \node at (-3.5,4) {\Large R};
     \node[style={fill=black,circle,draw,inner sep=1pt}] (u11) at (1.5,0) {};
     \node[style={fill=black,circle,draw,inner sep=1pt}] (u12) at (2.5,0) {};
     \node[style={fill=black,circle,draw,inner sep=1pt}] (u21) at (-1,0) {};
     \node[style={fill=black,circle,draw,inner sep=1pt}] (w12) at (2.5,3.5) {};
     \node[label=left:{$w_{2}$},style={fill=black,circle,draw,inner sep=1pt}] (w2) at (-2,0) {};
     \draw (w12) to (u11);
     \draw (w12) to (u12);
     \draw (w2) to (u21);
     \draw [->,red](2.5,4.3) to (2.5,3.8);
     \node (n1) at (2.5,4.7) {nice};
\end{scope}

\begin{scope}[shift={(12cm,-5cm)}]
    \draw(0,4) ellipse (4cm and 2cm);
    \draw (0,0) ellipse (4cm and 1cm);
    \node at (3.5,4) {\Large $F$ };
    \node (v6) at (0,-0.5){\Large $W$};
     \node (v7) at (-1,6.2) {};
     \node (v8) at (-1,1.8) {};
     \draw[bend left=20] (v7) to (v8);
      \node at (-3.5,4) {\Large R};
     \node[label=right:{$v$},style={fill=black,circle,draw,inner sep=1pt}] (v) at (0,0.5) {};
     \node[label=right:{$w$},style={fill=black,circle,draw,inner sep=1pt}] (w1) at (2,4.2) {};
     \node[style={fill=black,circle,draw,inner sep=1pt}] (w11) at (1.8,3) {};
     \node[style={fill=black,circle,draw,inner sep=1pt}] (w12) at (2.5,3.5) {};
     \node[style={fill=black,circle,draw,inner sep=1pt}] (u11) at (1.5,0) {};
     \node[style={fill=black,circle,draw,inner sep=1pt}] (u12) at (2.5,0) {};
     \node[label=right:{$w_{2}$},style={fill=black,circle,draw,inner sep=1pt}] (w2) at (0,4) {};
     \node[style={fill=black,circle,draw,inner sep=1pt}] (u21) at (-1,0) {};
      \draw (v) to (w1);
      \draw (w1) to (w11);
      \draw (w1) to (w12);
      \draw (w11) to (u11);
      \draw (w12) to (u12);
      \draw (v) to (w2);
      \draw (w2) to (u21);

\end{scope}

\begin{scope}[shift={(24cm,-5cm)}]
    \draw(0,4) ellipse (4cm and 2cm);
    \draw (0,0) ellipse (4cm and 1cm);
    \node at (3.5,4) {\Large $F$ };
    \node (v6) at (0,-0.5){\Large $W$};
     \node (v7) at (-1,6.2) {};
     \node (v8) at (-1,1.8) {};
     \draw[bend left=20] (v7) to (v8);
      \node at (-3.5,4) {\Large R};
     \node[label=right:{$v$},style={fill=black,circle,draw,inner sep=1pt}] (v) at (0,0.5) {};
     \node[label=right:{$w$},style={fill=black,circle,draw,inner sep=1pt}] (w1) at (2,4.2) {};
     \node[style={fill=black,circle,draw,inner sep=1pt}] (w11) at (1.3,0.5) {};
     \node[style={fill=black,circle,draw,inner sep=1pt}] (w12) at (2.3,0.5) {};
     \node[style={fill=black,circle,draw,inner sep=1pt}] (u11) at (1.5,0) {};
     \node[style={fill=black,circle,draw,inner sep=1pt}] (u12) at (2.5,0) {};
     \node[label=right:{$w_{2}$},style={fill=black,circle,draw,inner sep=1pt}] (w2) at (0,4) {};
     \node[style={fill=black,circle,draw,inner sep=1pt}] (u21) at (-1,0) {};
      \draw (v) to (w1);
      \draw (w1) to (w11);
      \draw (w1) to (w12);
      \draw (w11) to (u11);
      \draw (w12) to (u12);
      \draw (v) to (w2);
      \draw (w2) to (u21);
      \draw [->,red](2,4.9) to (2,4.4);
      \node (t1) at (2,5.2) {tent};
      \draw [->,red](0,4.7) to (0,4.2);
      \node (n1) at (0,5.0) {nice};
\end{scope}

\node (s1) at (4,6) {};
\node (t1) at (7.5,8) {};
\tikzset{>=triangle 90}
\draw [->,line width=0.4mm] (s1)--  node[rotate=30,xshift=0mm, yshift=3mm] {delete $v$} (t1);

\node (s2) at (3.5,-1.6) {};
\node (t2) at (8,-4.3) {};
\tikzset{>=triangle 90}
\draw [->,line width=0.4mm] (s2)--  node[rotate=-34,xshift=-1mm, yshift=3mm] {put $v$ into $W$} (t2);

\node (s3) at (16,8) {};
\node (t3) at (20,8) {};
\tikzset{>=triangle 90}
\draw [->,line width=0.4mm] (s3)--  node[xshift=0mm, yshift=4mm] {Reduction $2$}node[xshift=0mm, yshift=-4mm] {Reduction $6$} (t3);

\node (s4) at (16,-3) {};
\node (t4) at (20,-3) {};
\tikzset{>=triangle 90}
\draw [->,line width=0.4mm] (s4)--  node[xshift=0mm, yshift=4mm] {Reduction $7$}node[xshift=0mm, yshift=-4mm] {Reduction $6$} (t4);
\end{tikzpicture}

\end{center}
\caption{Observation $1$-$3$} \label{Obs1-3}
\end{figure}
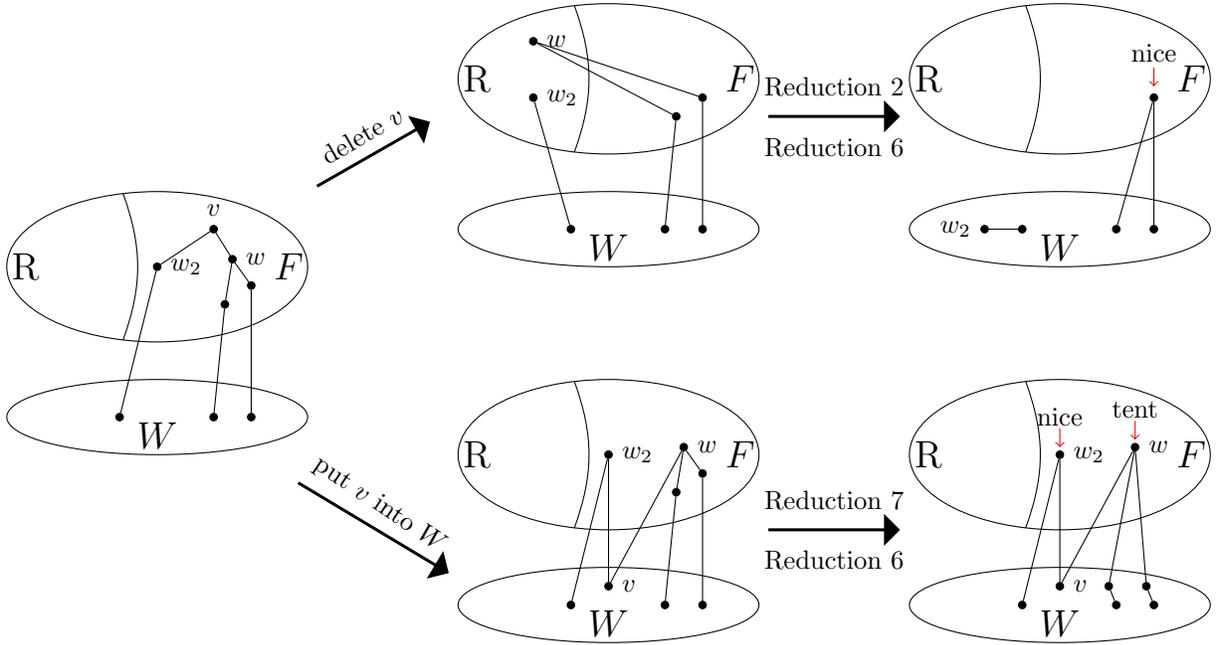
Armed with the above observations (see Fig.~\ref{Obs1-3}), we can now show that Reduction Rule 6 on its own
does not increase the measure.
\begin{claim}
An application of Reduction Rule 6 does not increase the measure.
\end{claim}
\begin{proof}
Since $v \in R$, $v$ is neither a tent nor a nice vertex.
Thus, moving $v$ to $W$ does not decrease $\eta$ nor $\tau$.

If $\deg_{W}(v)\geq 1$, $\rho$ does not increase by moving $v$ to $W$.
Hence, $\mu$ does not increase.

We are left with the case $\deg_{W}(v)=0$, and then moving $v$ to $W$ increases $\rho$ by one.
If $\Gdeg(v) \geq 1$ but $\deg_W(v)=0$, we have a P-nice neighbor $w$ of $v$.
Then, after $v$ is moved to $W$, Observation~\ref{obs:nice-W}
asserts that future application of reduction rules on $w$
cause a measure decrease of at least one, offsetting the increase of $\rho$.
Otherwise, $\Tdeg(v) \geq 1$, and we have a neighbor $w$ of $v$ that is a P-tent.
Then, after $v$ is moved to $W$, Observation~\ref{obs:tent-W}
asserts that future application of reduction rules on $w$ and its possible neighbors
in $F$ cause measure decrease of at least one. This finishes the proof.
 \end{proof}

%
%

\subsection{Branching for \textsc{DIS-IFVS}}\label{sec:Branching}

Now we are ready to introduce the branching algorithm.
We assume that all reduction rules have been applied exhaustively.
As a branching pivot, we pick a vertex $v \in F$ that is neither a nice vertex, nor a tent, nor a $P$-nice vertex,
and satisfies one of the following three cases:

\begin{description}
\item[Case A:] $\Gdeg(v)\geq 3$.
\item[Case B:] $\Gdeg(v)\geq 1$ and $\Tdeg(v)\geq 1$.
\item[Case C:] $\Tdeg(v)\geq 2$.
\end{description}

In case of more than one vertices of $F$ satisfying one of the above cases,
 we prefer to pick a vertex $v$ that satisfies an earlier case.
Within one case, we break ties arbitrarily.

First, note that the non-applicability of Reduction Rule 6 implies that the chosen
branching pivot $v$ does not lie in $R$.

No matter which case the chosen branching pivot $v$ satisfies, we branch into two cases.
In one case we include $v$ into the solution: we delete $v$ from the graph, include $N(v) \cap F$
into $R$, and decrease $k$ by one.
In the other case, we move $v$ to $W$.

We now show that in each of the cases, the branching gives a branching vector $(1,2)$ or better
with respect to the measure $\mu$. That is, in one of the branches the measure drops by at
least one, and in the other by at least two.

\medskip

\noindent \textbf{Case A:} $\Gdeg(v)\geq 3$.
\begin{enumerate}[(i)]
\item Branch where $v$ is deleted and all vertices in $N(v)\cap F$ are added to $R$.
$k$ decreases by $1$, $\rho$ stays the same, and $\eta$ and $\rho$ does not decrease as $v$ is neither a nice vertex nor a tent. Thus, $\mu$ decreases by at least one.
\item Branch where $v$ is moved from $F$ to $W$.
$\rho$ decreases by $\deg_{W}(v)-1$ (which may be $-1$ if $\deg_W(v)=0$)
and $\eta$ increases by $\Sdeg(v)$. Since $\Gdeg(v)=\deg_{W}(v)+\Sdeg(v)\geq 3$ and $\tau$ does not decrease,
$\mu$ decreases by at least two.
\end{enumerate}

\noindent \textbf{Case B:} $\Gdeg(v)\geq 1$ and $\Tdeg(v)\geq 1$.
\begin{enumerate}[(i)]
\item Branch where $v$ is deleted and all vertices in $N(v)\cap F$ are added to $R$.
First, $k$ decreases by one. Furthermore, $v$ has a P-tent neighbor $w$ and
Observation~\ref{obs:tent-X} asserts that future applications
of reduction rules on $w$ and its remaining neighbors in $F$ decrease the measure
by at least one. Thus, in total $\mu$ decreases by at least two.
\item Branch where $v$ is moved from $F$ to $W$.
For every P-tent neighbor $w$ of $v$, Observation~\ref{obs:tent-W} asserts that the application of reduction rules to $w$
and its remaining neighbors in $F$ cause a measure decrease of at least 1.
If $\deg_W(v) \geq 1$, then moving $v$ to $W$ does not increase $\rho$, and we are done.
Otherwise, if $\deg_W(v)=0$, moving $v$ to $W$ increases $\rho$ by $1$ but the assumption $\Gdeg(v) \geq 1$ implies that there also exists a P-nice neighbor $w$ of $v$.
For every such P-nice neighbor $w$ of $v$, Observation~\ref{obs:nice-W} asserts
that the future application of reduction rules on $w$ and its remaining neighbors in $F$
cause measure drop by at least $1$. Consequently, in this case we also have a measure drop of at least $1$.
\end{enumerate}

\noindent \textbf{Case C:} $\Tdeg(v)\geq 2$.
\begin{enumerate}[(i)]
\item Branch where $v$ is deleted and all vertices in $N(v)\cap F$ are added to $R$.
First, $k$ decreases by one.
Furthermore, for every P-tent neighbor $w$ of $v$, Observation~\ref{obs:tent-X}
asserts that the application of reduction rules on $w$ and its remaining neighbors in $F$
cause measure drop by at least one. Since $\Tdeg(v) \geq 2$, together with the decrease of $k$
we have a total measure decrease of at least $3$.
\item Branch where $v$ is moved from $F$ to $W$.
The move itself may increase $\rho$ by one.
For every P-tent neighbor $w$ of $v$, Observation~\ref{obs:tent-W} asserts
that the future application of reduction rules on $w$ and its remaining neighbors in $F$
cause measure drop by at least $1$.
Since $\Tdeg(v) \geq 2$, in total we have a measure decrease by at least $1$.
\end{enumerate}

We are left with analyzing what happens if no vertex of $F$ satisfies any of the three
cases for the choice of the branching pivot.
As in~\cite{ipec16}, we rely on the following base case.
\begin{lemma}[\cite{ipec16}] \label{poly}
Let $(G,W,R,k)$ be an instance of \textsc{DIS-IFVS} where every vertex in $V(G)\setminus W$ is either a nice vertex or a tent. Then we can find an independent feedback vertex set $X\subseteq V(G)\setminus (W\cup R)$ for $G$ of the minimum size in polynomial time.
\end{lemma}
Lemma \ref{poly} follows from the observation by Cao et al.~\cite{CaoC015} and the fact that all nice vertices and tents form an independent set.

We show the following.
\begin{lemma} \label{poly2}
If no reduction rule can be applied and every vertex of $F$ does not satisfy any
of the cases for the choice of the branching pivot, then the
remaining instance of \textsc{DIS-IFVS} can be solved in polynomial time.
\end{lemma}
\begin{proof}
We claim that every vertex in $F$ of the remaining graph $G$ is either a tent or a nice vertex;
the claim then follows by Lemma~\ref{poly}.

For contradiction, suppose that there is a connected component $D$
of $G[F]$ that is not a singleton with a tent or a nice vertex.
Since no vertex of $D$ falls into Case A, $\Gdeg(v) \leq 2$ for every $v \in D$.

Let $v \in D$ be any leaf, that is, a vertex in $F$ that has only exactly one neighbor in $F$.
If $\deg_W(v) = 0$, then $\deg(v) = 1$ and Reduction Rule~1 deletes $v$.
Thus, since $\Gdeg(v) \leq 2$, we have $\deg_W(v) \in \{1,2\}$.
In particular, every leaf of $D$ has at least one neighbor in $W$
and, since Reduction Rule~6 is not applicable to $v$, $v \notin R$.

Root the tree $G[D]$ at an arbitrary vertex, and consider a leaf $v \in D$
that is furthest from the root in $G[D]$ and, among such leaves, choose one
maximizing $\deg_W(v)$. Note that $v \notin R$ as otherwise Reduction Rule 6 would move $v$ to $W$.

First, assume $\deg_W(v) = 2$. Since $v$ is a leaf of $D$ and is not nice, $v$ has
exactly one neighbor $u \in D$, and $v$ is a P-tent.
Hence, $\Tdeg(u) \geq 1$.
If $\deg(u) \leq 1$, then Reduction Rule~1 applies to $u$.
If $\deg(u) = 2$, then Reduction Rule 7 applies to $v$ if $u \notin R$ and once
$u$ is in $R$, then Reduction Rule 6 applies to $u$, making $v$ a tent.
Consequently, $\deg(u) \geq 3$. However, by the choice of $v$, $\deg_W(u) \geq 1$ or $u$
is adjacent to another leaf $v'$ of $D$.
However, this implies that either 
\begin{itemize}
\item $\Gdeg(u) \geq 1$, if $\deg_W(u) \geq 1$ or $v'$ exists and $\deg_W(v') = 1$, i.e., $v'$ is P-nice; or
\item $\Tdeg(u) \geq 2$, if $v'$ exists and $\deg_W(v') = 2$, i.e., $v'$ is a P-tent.
\end{itemize}
Consequently, Case B or C applies to $u$.

Second, assume $\deg_W(v) = 1$, and again let $u$ be the unique neighbor of $v$ in $G[D]$.
If $\deg(u) = 2$, then Reduction Rule 2 is applicable.
By the choice of $v$, every other leaf $v'$ adjacent to $u$ also satisfies
$\deg_W(v') = 1$; that is, every child of $u$ is P-nice.
If $u \in R$, then Reduction Rule~6 applies to $u$.
If  $\Gdeg(u) \geq 3$, then Case A applies to $u$.
Hence, $u \notin R$, $\deg(u) = 3$, and $\Gdeg(u) = 2$: $u$ has a parent $x$ in $G[D]$
and either has one more child $v'$ that is P-nice or a neighbor in $W$.
In particular, $u$ is a P-tent, and $\Tdeg(x) \geq 1$.

If $\deg(x) = 2$, then Reduction Rule 7 would apply to $u$ and move $v$ to $R$, and consequently
Reduction Rule 6 would move $v$ to $W$.
If $\Gdeg(x) \geq 1$, then Case B applies to $x$.
Hence, $x$ has another child $u'$ that is not P-nice.
By the choice of $v$, the connected component of $G[D]\setminus \{x\}$ containing $u'$
is a star with $u'$ as a center.
Furthermore, since in the choice of $v$ we maximized $\deg_W(v)$, 
every child $w$ of $u'$ is P-nice (i.e., $\deg_W(w) = 1$).
Since Case A is not applicable to $u'$, we have $\Gdeg(u') \leq 2$.
If $\deg(u') = 2$, then either $u'$ is P-nice (if $\deg_W(u') = 1$) or Reduction Rule 2 is applicable to $u'$ and its child (if $\deg_W(u') = 0$).
We infer that $\deg(u') = 3$ and $\Gdeg(u') = 2$.
If $u' \in R$, then Reduction Rule~6 is applicable to $u'$.
We infer that $u'$ is a P-tent. Hence, $\Tdeg(x) \geq 2$ and Case C applies to $x$.
This completes the proof of the lemma.
 \end{proof}

Every step of the reduction rules and branching can be executed in polynomial time. In every case of branching, the branching vector is $(1,2)$. Thus we get the following recurrence:
$T(\mu)=T(\mu-1)+T(\mu-2)$.
As a result, the running time of the algorithm for DIS-IFVS is $O^{*}(\varphi^{2k})$.
This concludes the proof of Lemma~\ref{technical}
and thus of the whole Theorem~\ref{main}.

\section{Conclusion}\label{sec:Conclusion}
In this paper, we presented a faster FPT algorithm for the \textsc{Independent Feedback Vertex Set} problem by using a different measure, introducing some new reduction rules and improving the branching algorithm for the \textsc{Disjoint Independent Feedback Vertex Set} problem. Moreover, we introduce the notion of generalized degree and tent degree, which makes the reduction and branching more concise. The running time of our algorithm is $\Ohstar(3.619^{k})$, which matches the running time of the current fastest FPT algorithm for the \textsc{Feedback Vertex Set} problem. As \textsc{IFVS} is a more general problem than \textsc{FVS}, any improvement for \textsc{IFVS} will lead to an improvement for the FPT algorithm of \textsc{FVS}.
We conclude with re-iterating an open problem of~\cite{MisraPRSS12}:
does there exist a kernel of size $\Oh(k^{2})$, as it is the case for \textsc{FVS}~\cite{Iwata17,Thomasse10}?

\bibliography{IFVS-journal}

\begin{thebibliography}{10}

\bibitem{opt/2009}
{\em Encyclopedia of Optimization, Second Edition}.
\newblock Springer, 2009.

\bibitem{ipec16}
Akanksha Agrawal, Sushmita Gupta, Saket Saurabh, and Roohani Sharma.
\newblock Improved algorithms and combinatorial bounds for {I}ndependent
  {F}eedback {V}ertex {S}et.
\newblock In Jiong Guo and Danny Hermelin, editors, {\em 11th International
  Symposium on Parameterized and Exact Computation, {IPEC} 2016, August 24-26,
  2016, Aarhus, Denmark}, volume~63 of {\em LIPIcs}, pages 2:1--2:14. Schloss
  Dagstuhl - Leibniz-Zentrum fuer Informatik, 2016.
\newblock URL: \url{https://doi.org/10.4230/LIPIcs.IPEC.2016.2}, \href
  {http://dx.doi.org/10.4230/LIPIcs.IPEC.2016.2}
  {\path{doi:10.4230/LIPIcs.IPEC.2016.2}}.

\bibitem{AgrawalLMS18}
Akanksha Agrawal, Daniel Lokshtanov, Amer~E. Mouawad, and Saket Saurabh.
\newblock Simultaneous {F}eedback {V}ertex {S}et: {A} parameterized
  perspective.
\newblock {\em {TOCT}}, 10(4):18:1--18:25, 2018.
\newblock URL: \url{https://doi.org/10.1145/3265027}, \href
  {http://dx.doi.org/10.1145/3265027} {\path{doi:10.1145/3265027}}.

\bibitem{Bodlaender94}
Hans~L. Bodlaender.
\newblock On disjoint cycles.
\newblock {\em Int. J. Found. Comput. Sci.}, 5(1):59--68, 1994.
\newblock URL: \url{https://doi.org/10.1142/S0129054194000049}.

\bibitem{CaoC015}
Yixin Cao, Jianer Chen, and Yang Liu.
\newblock On {F}eedback {V}ertex {S}et: {N}ew measure and new structures.
\newblock {\em Algorithmica}, 73(1):63--86, 2015.
\newblock URL: \url{https://doi.org/10.1007/s00453-014-9904-6}, \href
  {http://dx.doi.org/10.1007/s00453-014-9904-6}
  {\path{doi:10.1007/s00453-014-9904-6}}.

\bibitem{ChenFLLV08}
Jianer Chen, Fedor~V. Fomin, Yang Liu, Songjian Lu, and Yngve Villanger.
\newblock Improved algorithms for {F}eedback {V}ertex {S}et problems.
\newblock {\em J. Comput. Syst. Sci.}, 74(7):1188--1198, 2008.
\newblock URL: \url{https://doi.org/10.1016/j.jcss.2008.05.002}, \href
  {http://dx.doi.org/10.1016/j.jcss.2008.05.002}
  {\path{doi:10.1016/j.jcss.2008.05.002}}.

\bibitem{pa-book}
Marek Cygan, Fedor~V. Fomin, Lukasz Kowalik, Daniel Lokshtanov, D{\'{a}}niel
  Marx, Marcin Pilipczuk, Michal Pilipczuk, and Saket Saurabh.
\newblock {\em Parameterized Algorithms}.
\newblock Springer, 2015.
\newblock URL: \url{https://doi.org/10.1007/978-3-319-21275-3}.

\bibitem{CyganNPPRW11}
Marek Cygan, Jesper Nederlof, Marcin Pilipczuk, Michal Pilipczuk, Johan M.~M.
  van Rooij, and Jakub~Onufry Wojtaszczyk.
\newblock Solving connectivity problems parameterized by treewidth in single
  exponential time.
\newblock In {\em {IEEE} 52nd Annual Symposium on Foundations of Computer
  Science, {FOCS} 2011, Palm Springs, CA, USA, October 22-25, 2011}, pages
  150--159. {IEEE} Computer Society, 2011.
\newblock URL: \url{https://doi.org/10.1109/FOCS.2011.23}.

\bibitem{CyganPP16}
Marek Cygan, Marcin Pilipczuk, and Michal Pilipczuk.
\newblock On {G}roup {F}eedback {V}ertex {S}et parameterized by the size of the
  cutset.
\newblock {\em Algorithmica}, 74(2):630--642, 2016.
\newblock URL: \url{https://doi.org/10.1007/s00453-014-9966-5}.

\bibitem{CyganPPW13}
Marek Cygan, Marcin Pilipczuk, Michal Pilipczuk, and Jakub~Onufry Wojtaszczyk.
\newblock {S}ubset {F}eedback {V}ertex {S}et is fixed-parameter tractable.
\newblock {\em {SIAM} J. Discrete Math.}, 27(1):290--309, 2013.
\newblock URL: \url{https://doi.org/10.1137/110843071}.

\bibitem{DowneyF92}
Rodney~G. Downey and Michael~R. Fellows.
\newblock Fixed parameter tractability and completeness.
\newblock In {\em Complexity Theory: Current Research, Dagstuhl Workshop,
  February 2-8, 1992}, pages 191--225. Cambridge University Press, 1992.

\bibitem{Fellows99}
Rodney~G. Downey and Michael~R. Fellows.
\newblock {\em Parameterized Complexity}.
\newblock Monographs in Computer Science. Springer, 1999.
\newblock URL: \url{https://doi.org/10.1007/978-1-4612-0515-9}.

\bibitem{Guillemot11a}
Sylvain Guillemot.
\newblock {FPT} algorithms for path-transversal and cycle-transversal problems.
\newblock {\em Discrete Optimization}, 8(1):61--71, 2011.
\newblock URL: \url{https://doi.org/10.1016/j.disopt.2010.05.003}.

\bibitem{GuoGHNW06}
Jiong Guo, Jens Gramm, Falk H{\"{u}}ffner, Rolf Niedermeier, and Sebastian
  Wernicke.
\newblock Compression-based fixed-parameter algorithms for {F}eedback {V}ertex
  {S}et and {E}dge {B}ipartization.
\newblock {\em J. Comput. Syst. Sci.}, 72(8):1386--1396, 2006.
\newblock URL: \url{https://doi.org/10.1016/j.jcss.2006.02.001}.

\bibitem{Iwata17}
Yoichi Iwata.
\newblock Linear-time kernelization for {F}eedback {V}ertex {S}et.
\newblock In {\em 44th International Colloquium on Automata, Languages, and
  Programming, {ICALP} 2017, July 10-14, 2017, Warsaw, Poland}, volume~80 of
  {\em LIPIcs}, pages 68:1--68:14. Schloss Dagstuhl - Leibniz-Zentrum fuer
  Informatik, 2017.
\newblock URL: \url{https://doi.org/10.4230/LIPIcs.ICALP.2017.68}.

\bibitem{IwataWY16}
Yoichi Iwata, Magnus Wahlstr{\"{o}}m, and Yuichi Yoshida.
\newblock Half-integrality, {LP}-branching, and {FPT} algorithms.
\newblock {\em {SIAM} J. Comput.}, 45(4):1377--1411, 2016.
\newblock URL: \url{https://doi.org/10.1137/140962838}.

\bibitem{KanjPS04}
Iyad~A. Kanj, Michael~J. Pelsmajer, and Marcus Schaefer.
\newblock Parameterized algorithms for {F}eedback {V}ertex {S}et.
\newblock In {\em Parameterized and Exact Computation, First International
  Workshop, {IWPEC} 2004, Bergen, Norway, September 14-17, 2004, Proceedings},
  volume 3162 of {\em Lecture Notes in Computer Science}, pages 235--247.
  Springer, 2004.
\newblock URL: \url{https://doi.org/10.1007/978-3-540-28639-4_21}.

\bibitem{KociumakaP14}
Tomasz Kociumaka and Marcin Pilipczuk.
\newblock Faster deterministic {F}eedback {V}ertex {S}et.
\newblock {\em Inf. Process. Lett.}, 114(10):556--560, 2014.
\newblock URL: \url{https://doi.org/10.1016/j.ipl.2014.05.001}.

\bibitem{KratschW12}
Stefan Kratsch and Magnus Wahlstr{\"{o}}m.
\newblock Representative sets and irrelevant vertices: New tools for
  kernelization.
\newblock In {\em 53rd Annual {IEEE} Symposium on Foundations of Computer
  Science, {FOCS} 2012, New Brunswick, NJ, USA, October 20-23, 2012}, pages
  450--459. {IEEE} Computer Society, 2012.
\newblock URL: \url{https://doi.org/10.1109/FOCS.2012.46}.

\bibitem{wg2018}
Shaohua Li and Marcin Pilipczuk.
\newblock An improved {FPT} algorithm for {I}ndependent {F}eedback {V}ertex
  {S}et.
\newblock In Andreas Brandst{\"{a}}dt, Ekkehard K{\"{o}}hler, and Klaus Meer,
  editors, {\em Graph-Theoretic Concepts in Computer Science - 44th
  International Workshop, {WG} 2018, Cottbus, Germany, June 27-29, 2018,
  Proceedings}, volume 11159 of {\em Lecture Notes in Computer Science}, pages
  344--355. Springer, 2018.
\newblock URL: \url{https://doi.org/10.1007/978-3-030-00256-5\_28}, \href
  {http://dx.doi.org/10.1007/978-3-030-00256-5\_28}
  {\path{doi:10.1007/978-3-030-00256-5\_28}}.

\bibitem{LokshtanovRS18}
Daniel Lokshtanov, M.~S. Ramanujan, and Saket Saurabh.
\newblock Linear time parameterized algorithms for {S}ubset {F}eedback {V}ertex
  {S}et.
\newblock {\em {ACM} Trans. Algorithms}, 14(1):7:1--7:37, 2018.
\newblock URL: \url{http://doi.acm.org/10.1145/3155299}.

\bibitem{MisraPRS12}
Neeldhara Misra, Geevarghese Philip, Venkatesh Raman, and Saket Saurabh.
\newblock On parameterized {I}ndependent {F}eedback {V}ertex {S}et.
\newblock {\em Theor. Comput. Sci.}, 461:65--75, 2012.
\newblock URL: \url{https://doi.org/10.1016/j.tcs.2012.02.012}.

\bibitem{MisraPRSS12}
Neeldhara Misra, Geevarghese Philip, Venkatesh Raman, Saket Saurabh, and
  Somnath Sikdar.
\newblock {FPT} algorithms for {C}onnected {F}eedback {V}ertex {S}et.
\newblock {\em J. Comb. Optim.}, 24(2):131--146, 2012.
\newblock URL: \url{https://doi.org/10.1007/s10878-011-9394-2}.

\bibitem{Thomasse10}
St{\'{e}}phan Thomass{\'{e}}.
\newblock A $4k^2$ kernel for {F}eedback {V}ertex {S}et.
\newblock {\em {ACM} Trans. Algorithms}, 6(2):32:1--32:8, 2010.
\newblock URL: \url{http://doi.acm.org/10.1145/1721837.1721848}.

\end{thebibliography}

\end{document}